\documentclass{eptcs}
\usepackage{breakurl}
\usepackage{latexsym}
\usepackage{graphicx}
\usepackage{amssymb,amsmath,amsthm}
\newcommand{\uj}{\newcommand}
\uj{\vege}{\hspace{1cm} \raisebox{-.6ex}{$\Box$}}
\uj{\vspp}{}
\uj{\vspe}{}
\uj{\vsp}{}
\uj{\vspee}{}
\uj{\vspm}{}
\uj{\vspmm}{}
\uj{\vspmmm}{}
\uj{\vspmini}{}
\uj{\vspminni}{}
\uj{\hspe}{\hspace*{1cm}}
\uj{\nyil}{\rightarrow}
\uj{\nnyil}{\Rightarrow}
\uj{\lny}{\!\downarrow}
\uj{\paral}{\parallel}
\uj{\set}{\mbox{\bf Set}}
\uj{\rel}{\mbox{\bf Rel}}
\uj{\dqt}{\mbox{\bf DQT}}
\uj{\dqtn}{\dqt _0}
\uj{\fdqtn}{F_{\mbox {{\scriptsize {\bf DQT}}}_0}}
\uj{\egy}{\mbox {{\bf 1}}}
\uj{\eps}{\epsilon }
\uj{\fdh}{\mbox{\bf FdHilb}}
\uj{\qt}{\cal{QT}}
\uj{\Complex}{\mathbb{C}}
\uj{\fny}{\uparrow}
\uj{\ffny}{{\Uparrow}}
\uj{\kit}{\fny_{k-i}}
\uj{\iso}{\mbox{\bf{Iso}}}
\uj{\bckl}{\backslash}
\uj{\calh}{\mathcal{H}}
\uj{\calm}{\mathcal{M}}
\uj{\caln}{\mathcal{N}}
\uj{\calk}{\mathcal{K}}
\uj{\call}{\mathcal{L}}
\uj{\calu}{\mathcal{U}}
\uj{\calv}{\mathcal{V}}
\uj{\calz}{\mathcal{Z}}
\uj{\alg}{Alg}
\uj{\daga}{A^{\dagger}}
\uj{\dagb}{B^{\dagger}}
\uj{\dagc}{C^{\dagger}}
\uj{\dagd}{D^{\dagger}}
\uj{\kapo}{\updownarrow}
\uj{\trab}{\mathrm{Tr}_{A,B}}
\uj{\truab}{\trab ^{U}}
\uj{\fbab}{\uparrow _{A,B}}
\uj{\fbuab}{\fbab ^U}
\uj{\calp}{\mathcal{P}}
\uj{\calt}{\mathcal{T}}
\uj{\calc}{\mathcal{C}}
\uj{\calg}{\mathcal{G}}
\uj{\calf}{\mathcal{F}}
\uj{\cala}{\mathcal{A}}
\uj{\calb}{\mathcal{B}}
\uj{\calgsig}{\calg (\Sigma )}
\uj{\calnt}{\vec{\calt}}
\uj{\tot}{{\bf T}}
\uj{\ibf}{{\bf i}}
\uj{\jbf}{{\bf j}}
\uj{\wbf}{{\bf w}}
\uj{\robf}{\rho }
\uj{\oast}{\circledast}
\uj{\biset}{\mbox {\bf Bset}}
\uj{\ima}{\mbox {\bf IMA}}
\uj{\sdcc}{S$^2$DC$^2$}
\uj{\bmc}{c}
\uj{\imm}{\mathcal {I}}
\uj{\alp}{\mathcal {A}}
\uj{\sdc}{\mathcal {S}}
\uj{\bind}{\mbox{\bf Ind}}
\uj{\cali}{\mathcal{I}}
\uj{\calgbar}{\vec{\calg}}
\uj{\barsig}{\bar{\Sigma }}
\uj{\hatdelta}{\hat{\delta }}
\uj{\nyilg}{\vec{G}}
\uj{\otc}{\otimes _{\calc }}
\uj{\intc}{Int(\calc )}
\uj{\intcn}{Int_0(\calc)}

\theoremstyle{theorem}
\newtheorem{theorem}{Theorem}
\newtheorem{lemma}[theorem]{Lemma}
\newtheorem{proposition}[theorem]{Proposition}
\newtheorem{corollary}[theorem]{Corollary}

\theoremstyle{definition}
\newtheorem{definition}[theorem]{Definition}

\renewenvironment{proof}{\noindent\emph{Proof.}}{\hfill\textsf{q.e.d.}\medskip}

\title{Quantum Turing automata}
\author{Mikl\'os Bartha
\institute{
Department of Computer Science, Memorial University of Newfoundland,
St.\ John's, NL, Canada}
\email{bartha@mun.ca}
}

\def\titlerunning{Quantum Turing automata}
\def\authorrunning{M.\ Bartha}
\begin{document}
\maketitle
\begin{abstract}
A denotational semantics of quantum Turing machines having a quantum control
is defined in the dagger compact closed category of finite dimensional Hilbert
spaces. Using the Moore-Penrose generalized inverse,
a new additive trace is introduced on the restriction
of this category to isometries, which trace is carried over to
directed quantum Turing machines as monoidal automata. The Joyal-Street-Verity $Int$
construction is then used to extend this structure to a reversible bidirectional one.
\end{abstract}
\section{Introduction}
In recent years, following the endeavors of Abramsky and Coecke to
express some of the basic quantum-mechanical concepts in an abstract
axiomatic category theory setting, several models have been
worked out to capture the semantics of quantum information protocols
\cite{abr1} and programming languages \cite{prakash,has,sel}. Concerning
quantum hardware, an algebra of automata which include both classical and
quantum entities has been studied in \cite{waltt}. In all of these works,
while the model could manipulate quantum data structures, the
actual control flow of the data was assumed to be necessarily classical.

The objective of the present paper is to show that the idea of a quantum control
is logically sound and feasible, and to provide a denotational style
semantics for quantum Turing machines having such a control. At the same time,
the rigid topological layout of Turing machines as a linear array of
tape cells is replaced by a flexible graph structure, giving rise to the
concept of Turing automata and graph machines as introduced in \cite{tur}.
By denotational semantics we mean that the changing of the tape
contents caused by the entire computation process is specified directly as a
linear operator, rather than just one step of this process.

 Our presentation will use the language of \cite{abr1,tra,dag},
but it will be specific to the concrete dagger compact closed category
$(\fdh ,\otimes)$ of finite dimensional Hilbert spaces at this time.
One can actually
read Section~4 separately as an interesting study in linear algebra,
introducing a novel application of the Moore-Penrose generalized
inverse of range-Hermitian operators by taking their Schur complement in certain
block matrix operators. This is the main technical contribution of the paper.
We believe, however, that the category theory contributions are
far more interesting and relevant. All of these results are around the
well-known Geometry of Interaction (GoI) concept introduced originally by Girard
\cite{gir} in the late 1980's as an interpretation of linear logic. The ideas,
however, originate from and are directly related to a yet earlier
work \cite{acta} by the author on the axiomatization of flowchart schemes,
where the traced monoidal category axioms first appeared in an algebraic
context. Our category theory contributions are as follows:

\begin{enumerate}
\item We introduce a total trace on the monoidal subcategory
of $(\fdh ,\oplus )$ defined by isometries, which has previously been sought by
others \cite{scott1,scott2}.

\item We explain the role of the $Int$ construction for
traced monoidal categories \cite{tra} in turning a computation process
bidirectional or reversible.

\item We capture the phenomenon in (ii) above
by our own concept ``indexed monoidal algebra'' \cite{tur1}, which is an equivalent
formalism for certain regular self-dual compact closed categories.
\end{enumerate}

 Due to space limitations we have to assume familiarity with some
advanced concepts in category theory, namely traced monoidal categories
\cite{tra}, compact closed categories \cite{cc}, and the $Int$
construction that links these two types of symmetric monoidal
categories \cite{mcl} to each other. For brevity, by a
monoidal category we shall mean a symmetric monoidal one
throughout the paper.
\vspmmm
\section{Traced and compact closed monoidal categories}
  The following definition of (strict) traced monoidal categories uses the
terminology of \cite{tra}. Trace (called feedback
in \cite{acta}) in a monoidal category $\calc $ with unit object
$I$, tensor $\otimes $, and symmetries $c_{A,B}:A\otimes B\nyil B\otimes A$
is introduced as a left trace, i.e., an operation
$\calc (U\otimes A,U\otimes B)\nyil \calc (A,B)$.

\begin{definition}
A {\em trace\/} for a monoidal category $\calc $ is a family of functions
     \[ \truab : \calc (U\otimes A,U\otimes B)\nyil \calc (A,B) \]
natural in $A$ and $B$, dinatural in $U$, and satisfying the following three axioms:
\begin{description}
\item {\em vanishing:}
     \[ \trab ^{I}(f)=f\; ,\;\;\; \trab ^{U\otimes V}(g)=
     \trab ^{V}(\mathrm{Tr}_{V\otimes A,V\otimes B}^{U}(g)); \]
\item {\em superposing:}
    \[ \truab (f)\otimes g=\mathrm{Tr}_{A\otimes C,B\otimes D}^{U}(f\otimes g),
      \mbox{\ where\ } g:C\nyil D ; \]
\item {\em yanking:}
    \[ \mathrm{Tr}_{U,U}^U(c_{U,U})=1_U. \]
\end{description}
\end{definition}
We use the word {\em sliding\/} as a synonym for dinaturality in $U$.
When using the term {\em feedback\/} for
trace, the notation $\mathrm{Tr}$ changes to $\uparrow $ or $\ffny $, and we
simply write $\mathrm{Tr}^U$ ($\fny ^U$, $\ffny ^U$) for $\mathrm{Tr}_{A,B}^U$ whenever $A$
and $B$ are understood from the context. The reason for using three different
symbols for trace is the different nature of semantics associated with these symbols.

As it is customary in linear algebra,
we shall use the symbols $I$ and $0$ as ``generic'' identity (respectively, zero)
operators, provided that the underlying Hilbert space is understood from the context.
As a further technical simplification we shall be working with the
strict monoidal formalism, even though the monoidal category of Hilbert
spaces with the usual tensor product is not strict. It is known, cf.\ \cite{mcl}, that
every monoidal category is equivalent to a strict one.

\begin{definition}
A monoidal category $\calc $ is {\em compact closed\/} (CC, for short) if every
object $A$ has a left adjoint $A^*$ in the sense that there exist morphisms
$d_A:I\nyil A^*\otimes A$ (the unit map) and $e_A:A\otimes A^*\nyil I$ (the
counit map) for which the two composites below result in the identity
morphisms $1_A$ and $1_{A^*}$, respectively.
\begin{eqnarray*}
 A &\!\!\!\!=&\!\!\!\! A\otimes I\nyil _{1_A \otimes d_A}A\otimes
(A^*\otimes A) = (A\otimes A^*)\otimes A\nyil _{e_A\otimes 1_A}
I\otimes A = A, \\
A^* &\!\!\!\!=&\!\!\!\! I\otimes A^*\nyil _{d_A\otimes  1_{A^*}}(A^*\otimes A)
\otimes A^* =  A^*\otimes (A\otimes A^*)\nyil _{1_{A^*}\otimes e_A}
A^*\otimes I = A^*.
\end{eqnarray*}
\end{definition}
As it is well-known,
every CC category admits a so called {\em canonical trace\/} \cite{tra} defined
by the formula
\[ \truab f =(d_{U}\otimes 1_A)\circ (1_{U^*}\otimes f)\circ (e_{U^*}\otimes 1_B). \]
Notice that we write composition of morphisms ($\circ $) in a left-to-right order,
avoiding the use of ``;'', which some may find more appropriate. We do so in order
to facilitate a smooth transition from composition to matrix product in Section~4.
In the formula of canonical trace above we have made the additional silent
assumption that the involution $()^{*}$ is strict, so that  $U^{**}=U$ holds
for each object $U$. As it is known from \cite{seely}, this assumption can also
be made without loss of generality.

Recall from \cite{dag} that a {\em dagger monoidal category\/} is a monoidal category
$\calc $ equipped with an involutive, identity-on-objects contravariant functor
$^{\dagger }:\calc ^{op}\nyil \calc $ coherently preserving the symmetric monoidal
structure as specified in \cite{dag}. A {\em dagger compact closed category\/} is
a dagger monoidal category that is also compact closed, and such that the diagram
in Figure \ref{fig1} commutes for all objects $A$.

\begin{figure}
\centering
\includegraphics[scale=.55]{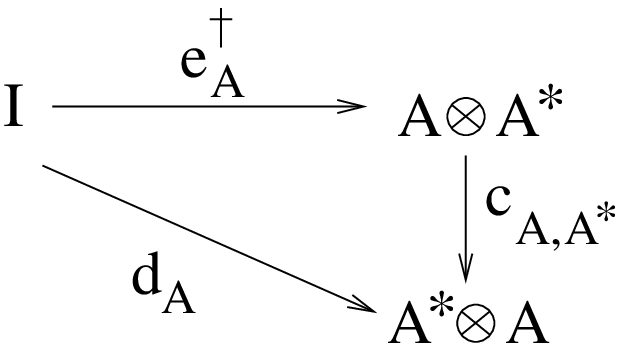}
\caption{Diagram for dagger compact closed categories}\label{fig1}
\end{figure}
\section{Monoidal vs.\ Turing automata}

Circuits and automata over an arbitrary monoidal category $M$ have been studied in
\cite{tcs,ic,sim,walt}. It was shown that the collection of such machines has the structure
of a monoidal category equipped with a natural feedback operation, which
satisfies the traced monoidal axioms, except for yanking. Moreover, sliding holds
in a weak sense, for isomorphisms only.

Let $A$ and $B$ be objects in $M$. An $M$-{\em automaton\/}
(circuit) $A\nyil B$ is a pair $(U,\alpha )$, where $U$ is a further object and $\alpha :
U\otimes A\nyil U\otimes B$ is a morphism in $M$.
If, for example, $M=(\set ,\times)$, then the pair $(U,\alpha )$
represents a deterministic Mealy automaton with states $U$, input $A$, and output $B$.
The structure of $M$-automata/circuits has been described as
a monoidal category $\mathrm{Circ}(M)$ with feedback in \cite{walt}.
This category was also shown to be freely generated by $M$.

In this paper we take a different approach to the study of monoidal automata.
We follow the method of \cite{tur} with the aim of constructing a {\em traced\/}
monoidal category as an adequate semantical structure
for these automata. One must not confuse this type of semantics with the meaning normally
associated with the category $\mathrm{Circ}(M)$ above, as they have seemingly very little in common.
A traced monoidal category indicates a {\em delay-free\/} semantics, as opposed to the
step-by-step {\em delayed\/} semantics suggested by $\mathrm{Circ}(M)$.
Moreover, the category that we are going to construct is not meant to be the quotient of
$\mathrm{Circ}(M)$ by the yanking identity, so as to turn it into a traced monoidal category
in the straightforward manner. Rather,
we define a brand new tensor and feedback (trace) on our $M$-automata, which are
analogous to the basic operations in iteration theories \cite{iter}. Regarding
the base category $M$, we shall assume an additional, so called additive tensor $\oplus $,
so that $\otimes $ distributes over $\oplus $. These two tensors will then be
``mixed and matched'' in the definition of tensor for $M$-automata,
providing them with an intrinsic Turing machine behavior.

The ``prototype'' of this construction, resulting in the CC category
of conventional Turing automata,
has been elaborated in \cite{tur1} using $M=(\rel ,\times ,+)$ as the base
category. This category was ideal as a template for the kind of construction we have
in mind, since it has a biproduct $+$ as the additive tensor and is self-dual compact
closed according to the multiplicative tensor $\times $.
Below we present the quantum counterpart of
this construction, working in the dagger compact closed category
of finite dimensional Hilbert spaces $(\fdh ,\otimes ,\oplus )$. More precisely,
the category $M$ above will
be the restriction of $\fdh $ to isometries as morphisms, which subcategory
is no longer compact closed and does not have a biproduct. \vspmmm
\section{Directed quantum Turing automata
}
In this section we present the construction outlined above, to obtain a strange
asymmetric model which does not yet qualify as a recognizable quantum computing device
in its own right. The model represents a Turing machine in which cells are interconnected
in a directed way, so that the control (tape head) always moves along interconnections
in the given fixed direction, should it be left or right. In other words, direction
is incorporated in the scheme-like  graphical syntax, rather than the semantics.
We use this model only as a stepping stone towards our real
objective, the (undirected) quantum Turing automaton described in Section~5.
\begin{definition}
A {\em directed quantum Turing automaton\/} is a quadruple \vspmini
   \[   T=(\calh ,\calk ,\call ,\tau ), \vspmmm \]
where $\calh $, $\calk $, and $\call $ are finite dimensional Hilbert spaces over
the complex field $\Complex $, and
   $\tau: \calh \otimes \calk \nyil \calh \otimes \call $
is an isometry in $\fdh $.
\end{definition}

Recall that an {\em isometry\/} between Hilbert spaces $\calh_1$ and $\calh _2$
is a linear map $\sigma :\calh _1\nyil \calh _2$ such that
$\sigma \circ \sigma ^\dagger =I$,
where $\sigma ^{\dagger}$ is the (Hilbert space) {\em adjoint\/} of $\sigma $.
Following the notation of general monoidal automata we write $T:\calk \nyil \call $,
and call the isometry $\tau $ the {\em transition operator\/} of $T$. Thus,
$T$ is the monoidal automaton $(\calh ,\tau ):\calk \nyil \call $. Sometimes we
simply identify $T$ with $\tau $, provided that the other parameters of $T$ are understood
from the context.

\begin{figure}[h]
\centering
\includegraphics[scale=.55]{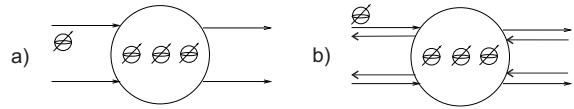}
\vspmmm
\caption{Two simple DQTA } \vspmmm \vspminni \label{fig2}
\end{figure}

The reader can obtain an intuitive understanding of the automaton $T$ from Figure \ref{fig2}a.
The state space $\calh $ is represented by a finite number of qubits (in our example 3),
while the control is a moving particle that moves from one of the input
interfaces (space $\calk $) to one of the output ones (space $\call $). It can
only move in the input $\nyil $ output direction, as specified by the operator $\tau $.
The number of input and output interfaces is finite. The control itself does not
carry any information, it is just moving around and changes the state of $T$.
In comparison with conventional Turing machines, the state of $T$ is the tape
contents of the
corresponding Turing machine, and the current state of the Turing machine is just an
interface identifier for $T$. For example, one can consider the DQTA in
Figure \ref{fig2}b as one tape cell of a Turing machine $TM$ having $2^3$ symbols in its tape
alphabet and only 2 states (2 left-moving and 2 right-moving interfaces, both input
and output). Correspondingly, $\calh $ is $8$-dimensional, while the dimension of
both $\calk $ and $\call $ is $4$. In motion,
if the control particle of $T$ resides on
the input interface labeled $(L,i)$ ($(R,i)$), then $TM$ is in state $i$ moving
to the left (respectively, right). The point is, however, that the automaton
$T$ need not represent just one cell, it could
stand for any finite segment of a Turing machine, in fact a Turing graph machine in the
sense of \cite{tur}. In our concrete example, a segment of $TM$ with $n$ tape cells
would have $3n$ qubits inside the circle of Figure \ref{fig2}b, but still the same $4+4$
interfaces.

An {\em isometric isomorphism\/}
$\sigma :\calh _1\nyil \calh _2$ ({\em unitary map}, if $\calh _1=\calh _2$)
is a linear operator such that both $\sigma $ and $\sigma ^{\dagger }$
are isometries. Two automata $T_i:(\calh _i,\tau _i):\calk \nyil \call $, $i=1,2$, are
{\em isomorphic}, notation $T_1\cong T_2$, if there exists an isometric isomorphism
$\sigma :\calh _1\nyil \calh _2$ for which \vspmini
\[ \tau _2 =(\sigma ^{\dagger }\otimes I_{\calk })
\circ \tau _1 \circ (\sigma \otimes I_{\call }). \vspmini \]
For simplicity, though, we shall work with
representatives, rather than equivalence classes of automata.

  Turing automata can be composed by the standard {\em cascade product\/} of monoidal
automata, cf.~\cite{ic,sim,walt}. If $T_1=(\calh _1,\tau _1): \call \nyil \calm $
and $T_2=(\calh_2, \tau _2):\calm \nyil \caln $ are directed quantum Turing automata
(DQTA, for short), then \vspmini
\[  T_1 \circ T_2 =(\calh _1\otimes \calh _2 ,\call ,\caln ,\tau ) \vspmini \]
is the automaton whose transition operator $\tau $ is \vspmini
\[ (\pi _{\calh _1,\calh _2}\otimes I_{\call })\circ (I_{\calh _2 }\otimes \tau _1)
 \circ (\pi _{\calh _2,\calh _1}\otimes I_{\calm })\circ (I_{\calh _1}\otimes \tau_2),
  \vspmini \]
where $\pi _{\calh ,\calk }$ is the symmetry $\calh \otimes \calk \nyil \calk \otimes \calh $
in $(\fdh ,\otimes )$.
As known from \cite{walt}, the cascade product of automata is compatible with
isomorphism, so that it is well-defined on isomorphism classes of DQTA.
The identity Turing automaton $1_{\calk }:\calk \nyil \calk $ has the
unit space $\Complex $ as its state space, and its transition operator is
simply $I_{\calk }$. The results in \cite{walt} imply that these data define a category
$\dqt $ over finite dimensional Hilbert spaces as objects, in which the morphisms are
isomorphism classes of DQTA.

Now let \vspmmm
\[  T_1=(\calh_1, \tau _1): \calk_1\nyil \call _1 \mbox{\ and\ }
   T_2=(\calh_2, \tau _2): \calk_2 \nyil \call _2 \vspmini
\]
be DQTA, and define $T_1\boxplus T_2$
to be the automaton over the state space $\calh _1\otimes \calh _2$ whose transition
operator \vspmini
  \[ \tau =\tau _1\boxplus \tau _2: (\calh_1\otimes \calh_2)\otimes (\calk_1\oplus
  \calk_2)\nyil(\calh_1\otimes \calh_2)\otimes (\call_1\oplus \call_2) \vspmini \]
acts as follows: $\tau \simeq \sigma _1\oplus \sigma _2$, where the morphisms \vspmini
   \[ \sigma _i :(\calh_1\otimes \calh_2)\otimes \calk _i\nyil (\calh_1\otimes \calh_2)
    \otimes \call_i ,\,\,\,\, i=1,2 \mbox{\ are:} \vspmmm \]
\[
 \sigma _1=(\pi _{\calh _1,\calh_2}\otimes I_{\calk_1})\circ (I_{\calh _2}\otimes
   \tau_1)\circ (\pi _{\calh _2,\calh_1}\otimes I_{\call_1}), \mbox{\ \ and\ \ }
   \sigma _2=I_{\calh _1}\otimes \tau _2.
\vspmini \]

In the above equations, $\oplus $ denotes the orthogonal sum of Hilbert spaces. Intuitively,
$\tau $ is the selective performance of {\em either\/} $\tau _1$ {\em or\/} $\tau _2$ on
the tensor space $\calh _1\otimes \calh_2 $. We say ``either or'', because the
interfaces of $T_1$ and $T_2$ are separated by $\oplus $, rather than $\otimes $.
The natural isomorphism $\simeq $ is
{\em distributivity\/} in the sense of \cite[Proposition~5.3]{abr1}.
It is clear that the operator $\tau_1\boxplus \tau_2$ is an isometry, so that the
operation $\boxplus $ is well-defined. We call this operation the {\em Turing tensor\/}.
The Turing tensor is also associative, up to natural isomorphism, of course.

The symmetries $\calk \oplus \call \nyil \call \oplus \calk $ associated with $\boxplus $
are the ``single-state'' Turing automata whose transition operator is the permutation
\vspmini
\[ \kappa _{\calk ,\call}=\!\!\!
\begin{array}{cc}
  \  & \begin{array}{cc}
       \call & \calk
       \end{array} \\
\begin{array}{c}
    \calk \\
    \call
\end{array}
& \!\!\!\!\!
\left(
\begin{array}{cc}
0 & I \\
I & 0
\end{array}
\right)
\end{array}
\!\!: (\Complex \otimes )(\calk \oplus \call )\nyil
  (\Complex \otimes )(\call \oplus \calk )
. \vspmmm \]
Along the lines of \cite{walt} it is routine to check that $\boxplus $ is also compatible
with isomorphism of automata,
and $(\dqt ,\boxplus )$ becomes a monoidal category in this way.

Our third basic operation on DQTA is feedback. Feedback follows the scheme of iteration
in Conway matrix theories \cite{iter}, using an appropriate star operation.  Let
$T:\calu \oplus \calk \nyil \calu \oplus \call $ be a DQTA having \vspmmm
\[ \tau : \calh \otimes (\calu \oplus \calk )\nyil \calh \otimes (\calu \oplus \call )
\vspmmm \]
as its transition operator. Then $\fny ^{\calu }T :\calk \nyil \call $ is the automaton
over (the {\em same\/} space) $\calh $ specified as follows. Consider the matrix of $\tau $:
\vspmm
\[
\begin{array}{cc}
  \  & \begin{array}{cc}
       \calh \otimes \calu & \calh \otimes \call
       \end{array} \\
\begin{array}{c}
    \calh \otimes \calu \\
    \calh \otimes \calk
\end{array}
& \!\!\!\!
\left(
\begin{array}{cccccc}
\ & \tau _A & & & \tau _B &\  \\
\ & \tau _C & & & \tau _D &\
\end{array}
\right)
\end{array}
\vspmmm \]
according to the biproduct decomposition \vspmmm
    \[ \tau = \langle [\tau_A ,\tau _C],[\tau _B,\tau _D ]\rangle ,\vspmmm \]
where
$[\_\,,\_\,]$ stands for coproduct and $\langle \_\,,\_\,\rangle $ for product.
The transition operator of $\fny ^{\calu }T$ is defined by the {\em Kleene
formula}:
\vspmini
\begin{equation}
  \fny ^{\calu }\tau =\lim _{n\nyil \infty} (\tau _D +\tau _C \circ\tau _A^{*n}
     \circ \tau _B) .\vspmini
\end{equation}
In the Kleene formula, $\tau _A^{*n}=\sum _{i=0}^n \tau _A^i$, where $\tau _A^0=I$
and $\tau _A^{i+1}=\tau _A^i\circ \tau_A$. In other words, $\tau _A^{*n}$ is the $n$-th
approximation of $\tau _A$'s {\em Neumann series\/} well-known in operator
theory.
The correctness of the above definition is contingent upon the existence of the
limit and also on the resulting operator being an isometry. For these two
conditions we need to make a short digression, which will also clarify the
linear algebraic background.

Let $\iso $ denote the subcategory of $\fdh $ having only isometries as its morphisms.
Notice that $(\iso ,\otimes)$ is no longer
compact closed, even though the multiplicative tensor $\otimes $ is still intact in it.
(The duals are gone.) This tensor, however, does not concern us at the moment.
Consider $\oplus $ as an additive tensor in $\iso $:
\vspmini
\[\tau _1\oplus \tau _2=\langle [\tau _1,0],[0,\tau _2]\rangle
 \mbox{\ \ for all isometries $\tau _i:\calh _i\nyil \calk _i$, $i=1,2$.} \vspmini \]
Clearly, $\tau _1\oplus \tau _2$ is an isometry. The new additive unit (zero) object
is the zero space $\calz $. With
the additive symmetries $\kappa _{\calh ,\calk }:\calh \oplus \calk \nyil \calk \oplus \calh $,
$(\iso ,\oplus)$ again qualifies as a monoidal category.
The biproduct property of $\oplus $ is lost, however. Nevertheless,
one may attempt to define a trace operation $\fny ^{\calu } \tau $ in $\iso $ by the Kleene
formula (1), where $\tau :\calu \oplus \calk \nyil \calu \oplus \call $.
(Cut $\calh \otimes $ in the matrix of $\tau $.)

Since the Kleene formula does not appear to be manageable, we first redefine $\fny ^\calu \tau $
and prove the equivalence of the two definitions later. Let \vspmmm
\begin{equation}
  \ffny ^\calu \tau =\tau _D +\tau _C\circ (I-\tau _A)^+\circ \tau _B, \vspmmm
\end{equation}
where $(\,)^+$ denotes the {\em Moore-Penrose generalized inverse\/}
of linear operators. Recall, e.g., from
\cite{isr} that the Moore-Penrose inverse (MP inverse, for short) of an arbitrary operator
$\sigma :\calh \nyil \calk $ is the unique operator $\sigma ^+:\calk \nyil \calh $ satisfying
the following two conditions:
\begin{enumerate}
\item    $\sigma \circ \sigma ^+ \circ \sigma =\sigma $, and $\sigma ^+\circ \sigma \circ \sigma ^+=
     \sigma ^+$;
\item   $\sigma \circ \sigma ^+$ and $\sigma ^+\circ \sigma $ are Hermitian.
\end{enumerate}

 The connection between formulas (1) and (2) is the following. If the Neumann series
$\tau _A^*$ converges, then $(I-\tau _A)$ is invertible and \vspmmm
\[ \tau _A^*=(I-\tau _A)^{-1}= (I-\tau _A)^+. \vspmmm\]
We know that $\|\tau _A\|\leq 1$, where $\|\,\|$ denotes the operator norm.
($\tau $ is an isometry.) Therefore the Kleene formula needs an explanation only if $\|\tau _A\|=1$.
In that case, even if $(I-\tau _A)$ is invertible, $\tau _A^*$ may not converge.

  Just as the Kleene formula in computer science, the expression on the right-hand side of
equation (2) is well-known and frequently used in linear algebra. For a block matrix \vspmmm
\[           M=\left( \begin{array}{cc}
                    A & B \\
                    C & D
                   \end{array}
               \right), \vspmmm
\]
where $A$ is square, the matrix $D-CA^+B$ is called the
{\em Schur complement\/} of $A$ on $M$, denoted $A/M$.
Cf., e.g., \cite{isr}. Observe that, under the assumption $\calk =\call$, \vspmmm
\[
\ffny ^{\calu }\tau =I-(I-\tau _A)/(I-\tau). \vspmmm
\]
For this reason we call $\ffny ^{\calu }\tau $ the {\em Schur I-complement\/} of $\tau _A$
on $\tau $, and write $\ffny ^{\calu }\tau =\tau _A\bckl \tau $.

\begin{theorem}\label{4.1}
The operator $\tau _A\bckl \tau $ is an isometry.
\end{theorem}

\begin{proof}
Isolate the kernel $\caln $ of $(I-\tau _A)$, and let $\calu _0$ be the
orthogonal complement \cite{lin} of $\caln $ on $\calu $. The matrix of $(I-\tau _A)$ in
this breakdown is \vspmm
\begin{equation}
         I-\tau _A=\!\!
\begin{array}{cc}
  \  & \begin{array}{ccc}
       \caln & \ & \calu _0
       \end{array} \\
\begin{array}{c}
    \caln \\
    \calu _0
\end{array}
&
\left(
\begin{array}{cc}
0 & 0 \\
-\tau_A^{\caln } & I-\tau _A^0
\end{array}
\right)
\end{array}. \vspmmm
\end{equation}
Put this matrix (rather, $I-(I-\tau _A)$) in the top left corner of $\tau $: \vspmmm \vspminni
\[
\begin{array}{cc}
  \  & \begin{array}{ccccc}
       \caln & & \calu _0 & & \call
       \end{array} \vspace*{1ex} \\
\begin{array}{c}
    \caln \vspace*{1ex} \\
    \calu _0 \vspace*{1ex} \\
    \calk
\end{array}
&
\left(
\begin{array}{ccccc}
I & & 0 & & \tau _B^{\caln } \vspace*{1ex}\\
\tau _A^{\caln } & & \tau _A^0 & & \tau _B^0 \vspace*{1ex}\\
\tau _C^{\caln } & & \tau _C^0 & & \tau _D
\end{array}
\right)
\end{array}. \vspmmm
\]
Since $\tau $ is an isometry (regardless of its concrete orthogonal representation
as a matrix operator),
all entries in the above block matrix with
superscript $\caln $ must be $0$. Consequently, $(I-\tau_A^0)$ is invertible and
$\tau _A\bckl \tau =\tau _A^0\bckl
\tau _0$, where $ \tau _0:\calu _0\oplus \calk \nyil \calu _0\oplus \call $
is the restriction of $\tau $ to the bottom right $2\times 2$ corner.
Indeed, \vspmmm
\[
\left( \begin{array}{ccc}
                    0 & & 0 \\
                    0 & & I-\tau _A^0
                   \end{array}
               \right)^{+} =
\left( \begin{array}{ccc}
                    0 & & 0\\
                    0 & & (I-\tau _A^0)^{-1}
                   \end{array}
               \right),
\vspmmm
\]
so that \vspmini
\[
\tau _C\circ (I-\tau _A)^+\circ \tau _B=\tau _C^0\circ (I-\tau _A^0)^{-1}\circ \tau _B^0.
\vspmini
\]
It turns out from the above discussion that $(I-\tau _A)$ is {\em group invertible\/}
and {\em range-Hermitian}, cf.~\cite{isr,bern}. Therefore the MP inverse of $(I-\tau _A)$
coincides with its Drazin inverse, which is the group generalized inverse of this operator.
Cf.\ again \cite{isr,bern}. It follows that we can assume, without loss of generality, that
$(I-\tau _A)$ is invertible. Note that (3) is only a unitary similarity,
therefore the sliding axiom is needed to make this argument correct. Cf.\ Theorem~\ref{4.4} below.
For better readability, replace the symbols $\tau _A$, $\tau _B$,
$\tau _C$, and $\tau _D$ by $A$, $B$, $C$, and $D$, respectively.
Furthermore, ignore the composition symbol $\circ $ as if we were dealing with ordinary
matrix product. Then we have: \vspmini
\[
   \left(
\begin{array}{cc}
      A & B \\
      C & D
\end{array}
\right)
\left(
\begin{array}{cc}
      \daga & \dagc\\
      \dagb & \dagd
\end{array}
\right)
=
\left(
\begin{array}{cc}
      I & 0 \\
      0 & I
\end{array}
\right).
\vspmini
\]
The following four matrix equations are derived: \vspmini
\begin{eqnarray}
A\daga +B\dagb &=&I, \\
A\dagc +B\dagd &=&0, \\
C\daga +D\dagb &=&0, \\
C\dagc +D\dagd &=&I. \vspmini
\end{eqnarray}
We need to show that \vspmmm
\[
  (D+C(I-A)^{-1}B)(\dagd +\dagb (I-\daga )^{-1}\dagc )=I. \vspmini \]
The product on the left-hand side yields: \vspmini
\[
  D\dagd +D\dagb (I-\daga )^{-1}\dagc +C(I-A)^{-1}B\dagd
    +C(I-A)^{-1}B\dagb (I-\daga )^{-1}\dagc .
\vspmini
\]
By (5) and (6) this is equal to: \vspmini
\[
 D\dagd -C\daga (I-\daga )^{-1}\dagc -C(I-A)^{-1}A\dagc
          +C(I-A)^{-1}B\dagb (I-\daga )^{-1}\dagc , \vspmini
\]
which is further equal to $D\dagd +CQ\dagc $, where \vspmini
\[ Q\!=\!
(I-A)^{-1}B\dagb (I-\daga )^{-1}\!-\daga (I-\daga )^{-1}\!-(I-A)^{-1}A.
\vspmini
\]
According to (7) it is sufficient to prove that $Q=I$.
A couple of equivalent transformations follow.
Multiply both sides of $Q=I$ by $(I-A)$ from the left:
\begin{eqnarray*}
   B\dagb (I-\daga )^{-1}-(I-A)\daga (I-\daga )^{-1}-A&=& I-A, \\
   B\dagb (I-\daga )^{-1}-(I-A)\daga (I-\daga )^{-1}&=&I.
\vspmmm \vspmini
\end{eqnarray*}
Multiply by $(I-\daga )$ from the right:
\begin{eqnarray*}
B\dagb -(I-A)\daga &=&I-\daga , \\
B\dagb +A\daga &=&I.
\vspmmm \vspminni
\end{eqnarray*}
The result is equation (4), which is given. The proof is now complete.
\end{proof}

\begin{lemma}\label{4.2}
Let $\tau :\calu \oplus \calv \oplus \calk \nyil \calu \oplus \calv \oplus \call $ be an
isometry defined by the matrix
\[
\left(
\begin{array}{cc}
     \mbox{{\Large $M$}}
&
   \begin{array}{c}
      B_1 \\
      B_2
   \end{array}
\\
\begin{array}{cc}
  C_1 & C_2
\end{array}
&
D
\end{array}
\right),
\mbox{\ where\ }
M=\left(
\begin{array}{cc}
P & Q \\
R & S
\end{array}
\right). \]
If $I-(P\bckl M)=I-(S+R(I-P)^+Q)$ is invertible, then
\[ \ffny ^{\calv }(\ffny ^{\calu }\tau )=\ffny ^{\calu\oplus\calv}\tau.\]
\end{lemma}

\begin{proof}
Using the kernel-on-top representation of operators as explained under
Theorem~\ref{4.1}, we can assume (without loss of generality) that
$I-P$ is also invertible. Then the statement follows from the Banachiewicz
block inverse formula \cite [Proposition~2.8.7]{bern}:
\[
\left(
\begin{array}{cc}
  A & B \\
  C & D
\end{array}
\right)^{-1} =
\left( \!\!\!\!
\begin{array}{cc}
  A^{-1}\!+\!A^{-1}B(D\!-\!CA^{-1}B)^{-1} & \!\!-A^{-1}B(D\!-\!CA^{-1}B)^{-1} \\
  -(D-CA^{-1}B)^{-1}CA^{-1}       & \!\!(D-CA^{-1}B)^{-1}
\end{array}
\!\!\!
\right)\!,
\]
using $A=I-P$, $B=-Q$, $C=-R$, and $D=I-S$. Computations are left to the reader.
\end{proof}

Note that the Banachiewicz formula does not hold true for the MP or the Drazin
inverse of the given block matrix when $A^{-1}$ and $(D-CA^{-1}B)^{-1}$ are
replaced on the right-hand side by $A^+$ and $(D-CA^{+}B)^+$, respectively,
even if one of these
square matrices is invertible. There are appropriate block inverse formulas
for generalized inverses, cf.~\cite{bern}, but these formulas are extremely
complicated and are of no use for us.
\begin{lemma}\label{4.3}
Let $\tau :\calu \oplus \calv \oplus \calk \nyil \calu \oplus \calv \oplus \call $ be an
isometry as in Lemma \ref{4.2}. If $P\bckl M=I$, then \vspmini
\[ \ffny ^{\calv }(\ffny ^{\calu }\tau )=\ffny ^{\calu\oplus\calv}\tau .
\vspmini
\]
\end{lemma}

\begin{proof}
Again, we can assume that $I-P$ is invertible. To keep the computation simple, let
$\calu $ and $\calv $ both be 1-dimensional. This, too, can in fact be assumed
without loss of generality, if one uses an appropriate induction argument. The
induction, however, can be avoided at the expense of a more advanced matrix
computation. Thus,
\[
\tau =\left(
\begin{array}{ccccc}
  p & & q & & u_1 \\
  r & & s & & u_2 \\
  v_1\lny & & v_2\lny & & D
\end{array}
\right),
\vspmini
\]
where $u_i$ and $(v_i\lny )$, $i=1,2$ are row and column vectors, respectively.
To simplify the computation even further, let the numbers $p,q,r,s$ be real.
The $2\times 2$ matrix $I-M$ is singular and range-Hermitian, therefore it is
Hermitian (only because the numbers are real, see \cite[Corollary 5.4.4]{bern}),
so that it must be of the form
\vspmini
\[
   I-M=\left(
\begin{array}{ccc}
  a & & b \\
  b & & b^2/a
\end{array}
   \right)
\vspmini
\]
for some real numbers $a,b$ with $a=1-p\neq 0$. Then
\vspmini
\[
    \ffny _{\calu }\tau = \left(
\begin{array}{ccc}
  c & & u \\
  v\lny & & D'
\end{array}
\right),
\vspmmm
\]
where $c=(1-b^2/a)+b^2/a=1$,
\vspmmm
\begin{eqnarray*}
 u &=& u_2-(b/a)\cdot u_1 , \\
  (v\lny )&=& (v_2\lny)-(b/a)\cdot (v_1\lny), \mbox{\ \ and} \\
  D'&=&D+(1/a)\cdot (v_1\lny)u_1 .
\vspmmm
\end{eqnarray*}
Since $c=1$, $u$ and $(v\lny )$ must be $0$. Consequently,
\vspmmm
\begin{equation}
   a\cdot u_2 =b\cdot u_1 \mbox{\ and\ } a\cdot (v_2\lny )=b\cdot (v_1\lny ) .
\vspmmm
\end{equation}
In order to calculate $(I-M)^+$, let $M'=S(I-M)S^{-1}$, where $S=S^{-1}$ is the unitary
matrix
\vspmmm
\[
   S=\frac{1}{d}\cdot
\left(
\begin{array}{ccc}
   -b & & a \\
    a & & b
\end{array}
\right), \;\;
d^2=a^2+b^2.
\vspmmm
\]
After a short computation,
\vspmmm
\[
  M'=\left(
\begin{array}{ccc}
  0 & & 0 \\
  0 & & d^2/a
\end{array}
\right).
\vspmmm
\]
It follows that:
\vspmmm
\[
  (I-M)^+=S\left(
\begin{array}{ccc}
  0 & & 0 \\
  0 & & a/d^2
\end{array}
\right)S, \mbox{\ \ and}
\]
\[
\ffny _{\calu \oplus \calv }\tau = D+
(v_1\lny ,v_2\lny )S\left(
\begin{array}{ccc}
  0 & & 0 \\
  0 & & a/d^2
\end{array}
\right) S \left(
\begin{array}{c}
  u_1 \\
  u_2
\end{array}
\right).
\]
Comparing this expression with
\vspmmm
\[
      \ffny _{\calv } (\ffny _{\calu }\tau )=D'=D+(1/a)\cdot (v_1\lny )u_1 ,\vspmini
\]
we need to prove that
\[  (v_1\lny ,v_2\lny )S\left(
\begin{array}{ccc}
  0 & & 0 \\
  0 & & a/d^2
\end{array}
\right) S \left(
\begin{array}{c}
  u_1 \\
  u_2
\end{array}
\right) = \frac{1}{a}\cdot (v_1\lny )u_1 .
\vspmini
\]
On the left-hand side we have:
\[
  (a/d^4)\cdot (a\cdot v_1\lny +b\cdot v_2\lny )(a\cdot u_1 +b\cdot u_2), \]
which indeed reduces to $(1/a)\cdot (v_1\lny)u_1$ by the help of (8). The proof is complete.
\end{proof}

\begin{theorem}\label{4.4}
The operation $\ffny ^{\calu }$ defines a trace for the monoidal category
$(\iso ,\oplus )$.
\end{theorem}

\begin{proof}
Naturality can be verified by a simple matrix computation, left to the reader.
Regarding the sliding axiom, we know from \cite[Lemma~2.1]{tra} that slidings of
symmetries suffice for all slidings in the presence of the other axioms.
Let therefore $\sigma :\calv \nyil \calu $ be an arbitrary symmetry (or permutation,
in general), and
$\tau :\calu \oplus \calk \nyil \calu \oplus \call $ be an isometry with
$\langle [A,B],[C,D]\rangle $ being the biproduct decomposition (matrix) of $\tau $.
Then, for the ``matrix'' $S$ of $\sigma $:\vspmini
\begin{eqnarray*}
  & & \!\!\!\!\!\!\!\!\!\!\!\!\!\!\!
  \ffny ^{\calv }((\sigma \oplus I)\circ \tau \circ (\sigma ^{-1}\oplus I)) \\
  \;\;\; & = & D+CS^{-1}(I-SAS^{-1})^{+}SB =
   D+CS^{-1}(SS^{-1}-SAS^{-1})^{+}SB  \\
 & = &   D+CS^{-1}(S(I-A)S^{-1})^{+}SB =
   D+CS^{-1}S(I-A)^+S^{-1}SB \\
 & = &  D+C(I-A)^+B=\ffny ^{\calu }\tau . \vspmmm
\end{eqnarray*}
In the above derivation we have used the obvious property
$ (SMS^{-1})^+= SM^+S^{-1}$
of the MP inverse. Remember that $\sigma $ is a permutation,
so that $\sigma ^{-1}=\sigma ^{\dagger }$.
Superposing and yanking
are trivial. Therefore the only challenging axiom is vanishing.

  Let $\tau :\calu \oplus \calv \oplus \calk \nyil \calu \oplus \calv \oplus \call $
be an isometry given by the matrix \vspmini
\[
\left(
\begin{array}{cc}
      M & B \\
      C & D
\end{array}
\right),
\mbox{\ where\ }
M=\left(
\begin{array}{cc}
P & Q \\
R & S
\end{array}
\right).
\vspmini
\]
We need to prove that
$
\ffny ^{\calv }(\ffny ^{\calu }\tau )=\ffny ^{\calu\oplus\calv}\tau $.
Again, without loss of generality, we can assume that $(I-P)$ is invertible and
\vspmini
\[
I-P\bckl M=
\left(
\begin{array}{cc}
0 & 0 \\
0 & S_0
\end{array}
\right),
\vspmini
\]
where $\calv =\caln \oplus \calv _0$ and $S_0:\calv_0\nyil \calv_0$ is invertible.
If $\caln $ is the zero space, so that $I-P\bckl M$ itself is invertible, then
the statement follows from Lemma\ref{4.2}. Otherwise \vspmini
\[  \ffny ^{\calv }(\ffny ^{\calu }\tau )=\ffny ^{\calv _0}(\ffny ^{\caln }(\ffny ^{\calu }
  \tau )).
\vspmini
\]
By Lemma \ref{4.3}, $\ffny ^{\caln }(\ffny ^{\calu }\tau )= \ffny ^{\calu \oplus \caln }\tau $,
and by Theorem \ref{4.1}, \vspmini
\[
\ffny ^{\calv _0}(\ffny ^{\calu \oplus \caln }\tau )= \ffny ^{\calu \oplus \caln \oplus
  \calv _0}\tau =\ffny ^{\calu \oplus \calv }\tau .\vspmini
 \]
The proof is now complete.
\end{proof}

  At this point the reader may want to check the validity of the Conway semiring
axioms \vspmini
\[   (ab)^* =a(ba)^*b+1 ,
   \;\;\; (a+b)^*=(a^*b)^*a^* \mbox{\ \ for all $a,b\in \Complex $, where}
\vspmmm \]
\[
     c^*=(1-c)^+=
\left\{
\begin{array}{ll}
    (1-c)^{-1} & \mbox{if $c\neq 1$} \\
    0              & \mbox{if $c=1$.}
\end{array}
\right. \vspmini
\]
Cf.\ \cite{iter}. Obviously, they do not hold, but they come very close. It may also occur to
the reader that the Schur $I$-complement defines a trace in the whole category
$(\fdh ,\oplus )$. Of course this is not true either, because the Banachiewicz
formula does not work for the MP inverse.

 In the recent paper \cite{scott2}, the authors introduced the so called
kernel-image trace as a partial trace \cite{scott1} on any additive category $\calc $.
Given a morphism $\tau :\calu \oplus \calk \nyil \calu \oplus \call $ in $\calc $
with a block matrix
\vspmini
\[ \tau =\langle [\tau _A,\tau _C],[\tau _B,\tau _D]\rangle
\vspmini
\]
as above, the {\em kernel-image trace\/} $\kit ^{\calu }\tau $ is defined if
both $\tau _B$ and $\tau _C$ factor through $(I-\tau _A)$, that is, there exist
morphisms $i:\calk \nyil \calu $ and $k:\calu \nyil \call $ such that
\vspmini
\[  \tau _C= i\circ (I-\tau _A) \mbox{\ \ and\ \ } \tau _B=(I-\tau _A)\circ k.
\vspmini
\]
Cf.\ Figure \ref{fig3}. In this case \vspmini
\[ \kit ^{\calu }\tau =\tau _D+\tau _C\circ k =\tau _D + i\circ \tau _B.
\vspmini
\]
It is easy to see that  $\kit ^{\calu }\tau $ is always defined if
$\tau $ is an isometry, and $\kit ^{\calu }\tau = \ffny ^{\calu }\tau $.
(Use the kernel-on-top transformation of $(I-\tau _A)$ as in Theorem~\ref{4.1}.)
Therefore $\kit ^{\calu }$ is totally defined on $(\iso ,\oplus )$ and it
coincides with $\ffny ^{\calu }$. Using \cite[Remark 3.3]{scott2} we thus
have an alternative proof of our Theorem~\ref{4.4} above.
\begin{figure}[h]
\centering
\vspmini
\includegraphics[scale=.4]{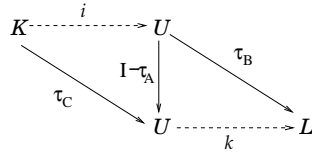}
\vspmini
\caption{The kernel-image trace} \vspmini \label{fig3}
\end{figure}

Now we turn back to the original definition of trace in $(\iso ,\oplus )$
by (1).

\begin{theorem}\label{4.5}
For every isometry $\tau :\calu \oplus \calk \nyil \calu \oplus \call $,
$\fny ^{\calu }\tau $ is well defined as an isometry $\calk \nyil \call $.
Moreover, \vspmini
\[ \fny ^{\calu }\tau =\ffny ^{\calu }\tau .\vspminni \]
\end{theorem}

\begin{proof}
This is in fact a simple formal language theory exercise. Take a concrete
representation of $\tau $ as an $(n+k)\times (n+l)$ complex matrix $(a_{ij})$,
where $n$, $k$, and $l$ are the dimensions of $\calu $, $\calk $, and $\call $,
respectively. For a corresponding set of variables $X=\{x_{ij}\}$, consider the
matrix iteration theory {\bf Mat }$\!_{L(X^*)}$
determined by the iteration semiring of all {\em formal power series\/} over
the $\omega $-complete Boolean semiring $\bf B$ with variables $X$
as described in Chapter~9 of \cite{iter}. The fundamental observation is
that $\fny ^n(a_{ij})$ is the evaluation of the series matrix $\fny ^n(x_{ij})$ under
the assignment $x_{ij}=a_{ij}$, provided that each entry in this matrix is
convergent. In our case, since $|a _{11}|\leq 1$, this matrix is definitely convergent
if $n=1$, and $\fny ^1(a_{ij})=\ffny ^1(a_{ij})$. A straightforward
induction on the basis of Theorem~\ref{4.4} then yields $\fny ^n(a_{ij})=\ffny ^n(a_{ij})$,
knowing that every iteration theory is a traced monoidal category.
\end{proof}

\begin{corollary}\label{4.6}
The monoidal category $(\dqt ,\boxplus )$ is traced by the feedback $\fny $.
\end{corollary}

\begin{proof}
Now the key observation is that, for every isometry $\tau :\calu \oplus \calk
\nyil \calu \oplus \call $ and object $\calm $,
\vspmini
\[
  (\ffny ^{\calu }\tau )\otimes I_{\calm }=\ffny ^{\calu \otimes \calm }(\tau \otimes I_{\calm}).
\vspmini
\]
This equation is an immediate consequence of
\vspmini
\[  (\sigma \otimes I)^+=\sigma ^+\otimes I ,
\vspmini
\]
which is an obvious property of the MP inverse. (Cf.\ the defining equations (i)-(ii) of
$\sigma ^+$.) In the light of this observation, each traced monoidal category axiom
is essentially the same in $(\dqt ,\boxplus )$ as it is in $(\iso ,\oplus )$.
Thus, the statement follows from Theorems~\ref{4.4} and~\ref{4.5}.
\end{proof}

\section{Making Turing automata bidirectional
}
Now we are ready to introduce the model of quantum Turing automata as a real
quantum computing device.
\vspe

\begin{definition}
A {\em quantum Turing automaton\/} (QTA, for short) of {\em rank\/} $\calk $ is a triple
$T=(\calh ,\calk , \tau )$, where $\calh $ and $\calk $ are finite dimensional
Hilbert spaces and $\tau :\calh \otimes \calk \nyil \calh \otimes \calk $
is a {\em unitary\/} morphism in $\fdh $.
\end{definition}

\begin{figure}[h]
\centering
\includegraphics[scale=.5]{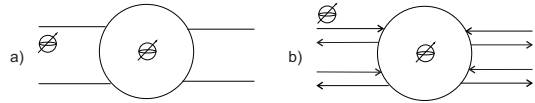}
\caption{One cell of a Turing machine as a QTA } \label{fig4}
\end{figure}
Again, two automata $T_i:(\calh _i, \calk ,\tau _i)$, $i=1,2$ are called
{\em isomorphic} if there exists an isometric isomorphism
$\sigma :\calh _1\nyil \calh _2$ for which $\tau _2 =(\sigma ^{\dagger }\otimes I_{\calk })
\circ \tau _1 \circ (\sigma \otimes I_{\calk })$.

{\bf Example.}\ \ In Figure \ref{fig4}a, consider the abstract representation of one tape cell
drawn from a hypothetical Turing machine having two states: $1$ and $2$. The tape alphabet
$\{ 0,1\}$ is also binary,
which means that there is a single qubit sitting in the cell. Thus, $\calh $ is 2-dimensional.
The control particle $c$ can reside on any of the given four interfaces. For example,
if $c$ is on the top left interface, then the control is coming from the left in state 1.
After one move, $c$ can again be on any of these four interfaces, so that the dimension
of $\calk $ is 4. Notice the undirected nature of one move, as opposed to the rigid
input$\to $output orientation forced on DQTA. The situation is, however, analogous to
having a separate input and dual output interface for each undirected one in
a corresponding DQTA. Cf.\ Figure \ref{fig4}b. The quantum Turing automaton obtained in this way
will then have a transition operator $\tau $ as an $8\times 8$ unitary matrix.
\vsp

Let $\calc $ be an arbitrary traced monoidal category. In order to describe the structure
of (undirected) quantum Turing automata we shall use a variant of the Joyal-Street-Verity
$Int$ construction \cite{tra} by which tensor is defined on objects in $Int(\calc)$ as
\vspmini
\[ (X,U)\otimes (X',U')=(X\otc X', U\otc U'),
\vspmini
\]
and on morphisms $f:(X,U)\nyil (Y,V)$, $f':(X',U')\nyil (Y',V')$ as
\vspmini
 \[ f\otimes f'=(1_X\otc c_{X',V}\otc 1_{V'})\circ (f\otc f')\circ
         (1_Y\otc c_{U,Y'}\otc 1_{U'}). \vspmini
 \]
Recall that $f:X\otimes V\nyil Y\otimes U$ in $\calc $. Correspondingly,
\vspmini
\[ 1_{(X,U)}=1_{X\otc U},\; c_{(X,U),(Y,V)}=c_{X,Y}\otc c_{V,U},\mbox{\ and\ }
  d_{(X,U)}=e_{(X,U)}=(c_{X,U})_{\calc}.
\vspmini
\]
The reason for the change is that, by
the original definition, the self-dual objects $(X,X)$ in $\intc $ are not closed
for the tensor.

\begin{definition}
A CC-category $\calc $ is {\em completely symmetric\/} if
$A=A^{**}$, $(A\otimes B)^*=A^*\otimes B^*$, and the natural isomorphism
$A^*\otimes B^*\!=\!(A\otimes B)^*\cong B^*\!\otimes A^*$ determined by the duality
$()^*$ coincides with $c_{A^*,B^*}$ for all objects $A,B$.
\end{definition}

In the above definition, ``the duality $()^*$'' refers to the pure autonomous
structure of $\calc $, forgetting the symmetries.
Observe that complete symmetry implies that the coherence conditions in effect for
the symmetries $c_{A,B}$ are automatically inherited by the units $d_A$ and
counits $e_A$ in an appropriate way, e.g.,
\[ d_{A^*}=d_A\circ c_{A^*,A} \mbox{\ \ and\ \ } d_{A\otimes B}=(d_A\otimes d_B)\circ
   (1_A\otimes c_{A,B^*}\otimes 1_B),
\vspmini
\]
as one would normally expect. These equations do not necessarily hold
without complete symmetry.

\begin{proposition}\label{P5.1}
For every traced monoidal category $\calc $, the CC-category
$\intc $ is completely symmetric.
\end{proposition}

\begin{proof}
Immediate by the definitions.\end{proof}

Let $\intcn $ denote the full subcategory of $\intc $ determined by its self-dual
objects $(X,X)$. Again, as an immediate consequence of the definitions, $()^*$ defines
a dagger structure on $\intcn $ through which it becomes a dagger compact closed category.
Clearly, the dagger (dual) of $f:X\otc Y\nyil Y\otc X$ as a morphism $(X,X)\nyil (Y,Y)$
is $c_{Y,X}\circ f\circ c_{Y,X}$. In general, we put forward the following
definition.

\begin{definition}
A {\em completely symmetric self-dual CC category\/} (\sdcc\ category,
for short) is a completely symmetric CC category such that $A=A^*$ for all objects $A$.
\end{definition}

\begin{corollary}
In every \sdcc\ category $\calc $, the contravariant functor $()^*$
defines a dagger structure on $\calc $ by which it becomes dagger compact closed.
Consequently, $d_A=d_A\circ c_{A,A}$ and $e_A=c_{A,A}\circ e_A$ hold in $\calc $.
For every traced monoidal category $\calc $, $\intcn $ is an \sdcc\ category.
\end{corollary}

\begin{proof}
Cf.\ Figure \ref{fig1}.
\end{proof}

Now let us assume that $\calc $ is a dagger traced monoidal category, that is,
$\calc $ has a monoidal dagger structure for which \vspmmm
\[    \mathrm{Tr}^U (f^{\dagger})=(\mathrm{Tr}^U f)^{\dagger}
        \mbox{\ \ for $f:U\otimes A\nyil U\otimes B$.} \vspmini \]
This is definitely the case for the subcategory $\dqtn $ of $\dqt $ consisting of automata
having an isometric isomorphism as their transition operator. Moreover, the map
$T\mapsto T\boxplus T^{\dagger}$ is injective in $\dqtn $.

\begin{theorem}\label{T5.3}
For every dagger traced monoidal category $\calc $, the map
$f\mapsto f\otimes f^{\dagger }$ defines a strict dagger-traced-monoidal functor
$F_{\calc }: \calc \nyil \intcn $ by which $F_{\calc }A=(A,A)$ for each object $A$.
\end{theorem}

\begin{proof}
Routine computation, left to the reader.\end{proof}

  At this point we have sufficient knowledge to understand the structure and behavior
of QTA. Indeed, any such automaton $(\calh ,\caln , \tau )$ with
$\tau :\calh \otimes \caln \nyil \calh \otimes \caln $
is in fact a morphism $(I,I)\nyil (\caln ,\caln )$
in the \sdcc\ category $Int_0(\dqtn)$. Using the terminology of \cite[Definition~3.2]
{abr1}, such a morphism is the {\em name\/} of any appropriate morphism
$(\calk ,\calk )\nyil (\call ,\call )$ in $Int_0(\dqtn )$
such that $\caln =\calk \oplus \call $. The
natural isomorphism induced by duality simply collapses these hom-sets into their
name hom-set. However, the reader should not be confused by the fact that
the name of a morphism $f:(X,X)\nyil (Y,Y)$ in $\intcn $
--- that is, $f: X\otimes Y\nyil Y\otimes X$ in $\calc $ --- is in fact a morphism
$X\otimes Y\nyil X\otimes Y$ in $\calc $, actually $f\circ c_{Y,X}$.

  In particular, for every automaton $T:\calk \nyil \call $ in $\dqtn $, the name of
$\fdqtn T=T\boxplus T^{\dagger }$ as a morphism $(I,I)\nyil (\calk \oplus
\call , \calk \oplus \call )$ is the QTA of rank $\calk \oplus \call $ which reflects
the joint behavior of $T$ and its reverse. Of course, however, the whole structure
of QTA is a lot richer than simply the image of $\dqtn $ under $\fdqtn $. This observation
is analogous to the obvious fact that the tensor of two vector spaces is richer than the
collection of tensors of individual vectors. Building on this analogy we can
consider the collection of QTA as a suitable algebraic structure, rather than
a category.

  An equivalent formalism for \sdcc\ categories in terms of so called indexed monoidal
algebras has been worked out in \cite{tur,tur1}. This new formalism deals with
QTA as ``vectors'' rather than morphisms, in the spirit explained in the previous
paragraph. The basis of the equivalence between indexed monoidal algebras and
\sdcc\ categories is the naming mechanism, which identifies morphisms with
their names. The advantage of using this algebraic framework is that it
simplifies the understanding of \sdcc\ categories by essentially collapsing
the dual category structure, which may sometimes be extremely but unnecessarily
convoluted.

\section{Conclusion}
We have provided a theoretical foundation for the study of quantum Turing machines
having a quantum control.
The dagger compact closed category $\fdh $ of finite dimensional Hilbert
spaces served as the basic underlying structure for this foundation. We
narrowed down the scope of this category to isometries, switched from multiplicative
to additive tensor, and defined a new additive trace operation by the help of
the Moore-Penrose generalized inverse. This trace was then
carried over to the monoidal category of directed quantum Turing automata.
Finally, we applied the $Int$ construction to obtain a compact closed category,
and restricted this category to its self-dual objects to arrive at our ultimate
goal, the model (indexed monoidal algebra) of undirected quantum Turing automata.
\vspmmm
\bibliographystyle{eptcs}

\begin{thebibliography}{11}

\providecommand{\urlalt}[2]{\href{#1}{#2}}
\providecommand{\doi}[1]{doi:\urlalt{http://dx.doi.org/#1}{#1}}


\bibitem{abr1} S.\ Abramsky \& B.\ Coecke (2004):
{\em A categorical semantics of quantum protocols}. In:
19th IEEE Symposium on Logic in Computer Science (LICS 2004), 14-17 July 2004, Turku, Finland, Proceedings
IEEE Computer Society Press, pp.\ 415--425, \doi{10.1109/LICS.2004.1319636}.

\bibitem{acta} M.\ Bartha (1987):
{\em A finite axiomatization of flowchart schemes}. {\sl Acta Cybernetica} 8, pp.\ 203--217.
%
\bibitem{tcs} M.\ Bartha (1987):
{\em An equational axiomatization of systolic systems}.
{\sl Theoretical Computer Science} 55, pp.\ 265--289, \doi{10.1016/0304-3975(87)90104-6}.
%
\bibitem{ic} M.\ Bartha (1992):
{\em An algebraic model of synchronous systems}.
{\sl Information and Computation} 97, pp.\ 97--131, \doi{10.1016/0890-5401(92)90006-2}.

\bibitem{sim}
M.\ Bartha (2008):
{\em Simulation equivalence of automata and circuits}.
In:
E.\ Csuhaj-Varj\'{u}, Z.\ \'{E}sik (eds.),
Automata and Formal Languages, 12th International Conference, AFL 2008, Balatonfüred, Hungary, May 27-30, 2008, Proceedings, pp.\ 86-99.

\bibitem{tur}
M.\ Bartha (2010):
{\em Turing automata and graph machines}.
In: S.\ B.\ Cooper, P.\ Panangaden, E.\ Kashefi (eds.), Proceedings Sixth Workshop on Developments in Computational Models: Causality, Computation, and Physics, Electronic Proceedings in Theoretical Computer Science 26,
pp.\ 19--31, \doi{10.4204/EPTCS.26.3}.

%
\bibitem{tur1}
M.\ Bartha (2013):
{\em The monoidal structure of Turing machines}.
{\sl Mathematical Structures in
Computer Science} 23(2):204-246, \doi{10.1017/S0960129512000096}.

\bibitem{isr}
A.\ Ben-Israel \& T.N.E.\ Greville (2003):
{\em Generalized Inverses:
Theory and Applications}.
Springer-Verlag, Berlin.
%
\bibitem{bern}
D.S.\ Bernstein (2005):
{\em Matrix Mathematics}. Princeton University Press,
Princeton, NJ.
%
\bibitem{iter} S.L.\ Bloom \& Z.\ \'Esik (1993):
{\em Iteration Theories: The
Equational Logic of Iterative Processes}.
Springer-Verlag, Berlin.
%
\bibitem{seely}
J.R.B.\ Cockett, M.\ Hasegawa, \& R.A.G.\ Seely (2006):
{\em Coherence of the
double involution on $*$-autonomous categories}.
{\sl Theory and Application of
Categories} 17(2), pp.\ 17--29.
%
\bibitem{prakash}
E.\ D'Hondt \& P.\ Panangaden (2006):
{\em Quantum weakest preconditions}.
{\sl Mathematical Structures in Computer Science} 16,
pp.\ 429--451, \doi{10.1017/S0960129506005251}.
%
\bibitem{waltt}
L.\ de Francesco Albasini, N.\ Sabadini, \& R.F.C.\ Walters
(2011):
{\em An algebra of automata which includes both classical and quantum entities}. {\sl Electronic Notes in Theoretical
Computer Science} 270, pp.\ 263--272, \doi{10.1016/j.entcs.2011.01.036}.

%
\bibitem{gir}
J.-Y.\ Girard (1989):
{\em Geometry of Interaction I: Interpretation of system F}. In:
R. Ferro, C. Bonotto, S. Valentini, A. Zanardo (eds.), Logic Colloquium '88. Proceedings of the colloquium held at the University of Padova, Padova, August 22–31, 1988,  Studies in Logic and the Foundations of Mathematics, 127. North-Holland, Amsterdam, pp.\ 221-260, \doi{10.1016/S0049-237X(08)70271-4}.

\bibitem{scott1}
E.\ Haghverdi \& P.J.\ Scott (2010):
{\em Towards a typed geometry of interaction}.
{\sl Mathematical Structures in Computer Science} 20,
pp.\ 1--49, \doi{10.1017/S096012951000006X}.

\bibitem{has}
I.\ Hasuo \& N.\ Hoshino (2011):
{\em Semantics of higher order quantum computation via
Geometry of Interaction}.
In:  Proceedings of the 26th Annual IEEE Symposium on Logic in Computer Science, LICS 2011, June 21-24, 2011, Toronto, Ontario, Canada, IEEE Computer Society Press,
pp.\ 237--246, \doi{10.1109/LICS.2011.26}.
%
\bibitem{tra}
A.\ Joyal, R.\ Street, \& D.\ Verity (1996):
{\em Traced monoidal categories}.
{\sl Mathematical Proceedings of the Cambridge Philosophical Society} 119, pp.\ 447--468, \doi{10.1017/S0305004100074338}.
%
\bibitem{walt} P.\ Katis, N.\ Sabadini, \&
R.F.C.\ Walters (2002):
{\em Feedback, trace, and fixed-point semantics}.
{\sl RAIRO Theoretical Informatics and Applications} 36, pp.\ 181--194, \doi{10.1051/ita:2002009}.

\bibitem{cc}
G.M.\ Kelly \& M.L.\ Laplaza (1980):
{\em Coherence for compact closed categories}.
{\sl Journal of Pure and Applied Algebra} 19, pp.\ 193--213, \doi{10.1016/0022-4049(80)90101-2}.
%
\bibitem{mcl} S.\ Mac Lane (1997):
{\em Categories for the Working Mathematician}, Springer-Verlag, Berlin.

\bibitem{scott2}
O.\ Malherbe, P.J.\ Scott \& P.\ Selinger (2011):
{\em Partially traced
categories}. {\sl Journal of Pure and Applied Algebra} 216(12),
pp.\ 2563--2585, \doi{10.1016/j.jpaa.2012.03.026}.

%
\bibitem{lin} S.\ Roman (2005):
{\em Advanced Linear Algebra},
Springer-Verlag, Berlin.
%
\bibitem{dag} P.\ Selinger (2007):
{\em Dagger compact closed categories and completely positive maps}.
{\sl Electronic Notes in Theoretical Computer Science} 170, pp.\ 139--163, \doi{10.1016/j.entcs.2006.12.018}.

\bibitem{sel} P.\ Selinger (2004):
{\em Towards a quantum programming language}.
{\sl Mathematical
Structures in Computer Science} 14, pp.\ 527--586, \doi{10.1017/S0960129504004256}.
\end{thebibliography}

\end{document}